\documentclass[aps,prl,reprint,letterpaper,showpacs]{revtex4-1}
\usepackage{dsfont}
\usepackage{amsmath}
\usepackage{amssymb}
\usepackage{amsthm}
\usepackage{graphicx}

\newcommand{\abs}[1]{\ensuremath{|#1|}}

\newcommand{\cancel}[1]{}
\newcommand{\norm}[2]{\ensuremath{|\!|#1|\!|_{#2}}}

\newcommand{\tr}{\textnormal{tr}}

\newcommand{\ptr}[1]{\textnormal{tr}_{\textnormal{\tiny #1}}}
\newcommand{\ptrace}[2]{\ensuremath{\ptr{#1} (#2)}}

\newcommand{\idx}[2]{{#1}_{#2}}

\newcommand{\ket}[1]{| #1 \rangle}
\newcommand{\keti}[2]{| #1 \rangle_{\textnormal{\tiny #2}}}
\newcommand{\bra}[1]{\langle #1 |}

\newcommand{\proj}[2]{| #1 \rangle\!\langle #2 |}
\newcommand{\proji}[3]{| #1 \rangle\!\langle #2 |_{\textnormal{\tiny #3}}}

\newcommand{\vecstate}[1]{\ket{#1}\bra{#1}}

\newcommand{\kron}{\otimes}
\newcommand{\eps}{\varepsilon}

\newcommand{\h}{\ensuremath{\mathcal{H}}}
\newcommand{\hi}[1]{\ensuremath{\mathcal{H}_{\textnormal{\tiny #1}}}}
\newcommand{\hA}{\hi{A}}
\newcommand{\hB}{\hi{B}}

\newcommand{\hAB}{\hi{AB}}

\newcommand{\id}{\ensuremath{\mathds{1}}}
\newcommand{\idi}[1]{\ensuremath{\mathds{1}_{\textnormal{\tiny #1}}}}
\newcommand{\idA}{\idi{A}}
\newcommand{\idB}{\idi{B}}

\newcommand{\idZ}{\idi{Z}}
\newcommand{\idAB}{\idi{AB}}

\newcommand{\idBC}{\idi{BC}}

\newcommand{\posops}[1]{\ensuremath{\mathcal{P}(#1)}}

\newcommand{\normstates}[1]{\ensuremath{\mathcal{S}_{=}(#1)}}
\newcommand{\subnormstates}[1]{\ensuremath{\mathcal{S}_{\leq}(#1)}}

\newcommand{\rhot}{\ensuremath{\tilde{\rho}}}
\newcommand{\rhob}{\ensuremath{\bar{\rho}}}

\newcommand{\rhoA}{\ensuremath{\idx{\rho}{A}}}
\newcommand{\rhoB}{\ensuremath{\idx{\rho}{B}}}

\newcommand{\rhoAB}{\ensuremath{\idx{\rho}{AB}}}
\newcommand{\rhoABZ}{\ensuremath{\idx{\rho}{ABZ}}}

\newcommand{\rhotAB}{\ensuremath{\idx{\rhot}{AB}}}

\newcommand{\rhoABC}{\ensuremath{\idx{\rho}{ABC}}}

\newcommand{\rhotABC}{\ensuremath{\idx{\rhot}{ABC}}}

\newcommand{\sigmat}{\ensuremath{\tilde{\sigma}}}

\newcommand{\sigmaB}{\ensuremath{\idx{\sigma}{B}}}

\newcommand{\tauA}{\ensuremath{\idx{\tau}{A}}}

\newcommand{\phiAB}{\ensuremath{\idx{\phi}{AB}}}

\newcommand{\psiAB}{\ensuremath{\idx{\psi}{AB}}}

\newcommand{\epst}{\ensuremath{\tilde{\eps}}}

\newcommand{\mX}{\mathcal X}

\newcommand{\mF}{\mathcal F}

\newcommand{\mZ}{\mathcal Z}

\newcommand{\mE}{\mathcal E}

\newcommand{\cI}{\mathcal{I}}

\newcommand{\cM}{\mathcal{M}}

\newcommand{\chh}[5]{\ensuremath{H_{#1}^{#2}({#3}|{#4})_{#5}}}

\newcommand{\chmin}[3]{\chh{\textnormal{min}}{}{#1}{#2}{#3}}
\newcommand{\chmineps}[4]{\chh{\textnormal{min}}{#1}{#2}{#3}{#4}}
\newcommand{\epsball}[2]{\ensuremath{\mathcal{B}^{#1}(#2)}}

\theoremstyle{plain}
\newtheorem{lemma}{Lemma}
\newtheorem{theorem}[lemma]{Theorem}

\theoremstyle{definition}
\newtheorem{definition}[lemma]{Definition}

\newtheorem*{definition2}{Definition}

\newcommand{\duni}[3]{\Delta(#1 | #2)_{#3}}
\newcommand{\dis}{D}

\begin{document}

\title{Impossibility of Growing Quantum Bit Commitments}

\author{Severin \surname{Winkler}} 
\affiliation{Computer Science Department, ETH Zurich, 8092 Zurich, Switzerland}
\author{Marco \surname{Tomamichel}}
\author{Stefan \surname{Hengl}}
\author{Renato \surname{Renner}}
\affiliation{Institute for Theoretical Physics, ETH Zurich, 8093 Zurich, Switzerland}

\begin{abstract}
  Quantum key distribution (QKD) is often, more correctly, called key
  growing. Given a short key as a seed, QKD enables two parties,
  connected by an insecure quantum channel, to generate a secret key
  of arbitrary length. Conversely, no key agreement is possible
  without access to an initial key. Here, we consider another
  fundamental cryptographic task, commitments. While, similar to key
  agreement, commitments cannot be realized from scratch, we ask
  whether they may be grown. That is, given the ability to commit to a
  fixed number of bits, is there a way to augment this to commitments
  to strings of arbitrary length? Using recently developed
  information-theoretic techniques, we answer this question in the
  negative.
\end{abstract}

\maketitle

\emph{Introduction.}|\, Quantum key distribution~\cite{BB84,E91}
allows two honest parties, Alice and Bob, to establish a shared secret
key, using only insecure quantum communication. However, a necessary
precondition for this to be possible is that they have access to a
pre-shared initial key, to be used for authentication|a fact that is
sometimes overlooked in the literature. It is easy to see that without
such an initial key, it is impossible for Alice to distinguish between
Bob and an eavesdropper pretending to be Bob\,|\,rendering all further
security considerations futile. Nevertheless, once an initial key is
available, this key can be \emph{grown}, i.e., expanded to arbitrary
length~\footnote{An explicit calculation that shows that a
  constant-length initial key is sufficient to generate arbitrarily
  many novel key bits is given, for example, in~\cite{muller09}.}.

Another similar example is coin tossing. It is known that there is no
unconditionally secure two-party protocol that generates a fair random
coin which cannot be biased by a dishonest
party~\cite{Blum82}. However, if the two parties have access to a
certain number of ideal coin tosses to start with, they can use
protocols to obtain a larger number of secure coin tosses. (Here,
security holds in a standalone model, where it is assumed that the
protocol is invoked only once~\cite{HMU06}.)

Following this line of thought, one may wonder whether other
cryptographic primitives, such as commitments~\cite{Blum82},
can be grown in a similar way. A \emph{string commitment} protocol
allows a sender to commit to a bit string that is revealed to a
receiver at a later point. The protocol is secure for the sender
(\emph{hiding}) if the receiver cannot gain information about the
commitment before she reveals it and it is secure for the receiver
(\emph{binding}) if the sender cannot change the string once
committed. Here, we are only interested in unconditionally secure
protocols, i.e., protocols that are secure against dishonest parties
with unlimited computing power.

While it is known that unconditionally secure commitments cannot be
implemented using classical or quantum communication
only~\cite{Mayers97,LoChau97} (see
also~\cite{DKSW06,CDPSW09}), this Letter strives to answer the
question whether it is possible to implement a long string commitment
with a protocol that uses a smaller number of bit commitments that are
provided as a resource. (A \emph{bit commitment} is a string commitment of
length one.)  We will answer this question to the negative, showing
that it is impossible to expand commitments even minimally, and even
under relaxed security criteria.

Commitments have a wide variety of applications in theoretical cryptography,
ranging from zero-knowledge proofs~\cite{GMR85} to secure coin tossing. In particular, commitments can be used to
implement statistically secure and universally composable oblivious
transfer~\cite{BBCS92,DFLSS09,Unruh10}, a functionality that is
sufficient to realize universal secure two-party
computation~\cite{Kilian88}.

In~\cite{WW10} it has been shown that unconditionally secure oblivious
transfer cannot be extended using quantum protocols. We note that this
already imposes certain bounds on the resources that can be obtained
from a limited number of bit commitments~\footnote{Using the
  equivalence of oblivious transfer and commitments, the result
  of~\cite{WW10} implies that there exists no \emph{composable}
  protocol that implements $(m+1)$ \emph{individual} bit commitments
  using $m$ bit commitments as a resource, if one demands that the
  error decreases exponentially in $m$.}.  Furthermore, bounds on the
quality of commitments for relaxed security definitions have been
shown in~\cite{SR01,BCHLW06,CK11}. Conversely, it has been shown that 
secure commitments can be implemented in relativistic settings involving 
multiple sites \cite{Kent11} or using trusted resources such as a noisy 
channel~\cite{Crepea97} or (trusted) distributed randomness~\cite{IMNW04,WolWul04}.

We now proceed with a more detailed specification of string commitment as well 
as the class of protocols we consider. We then briefly review the smooth
entropy calculus, which is required for our technical arguments. Our
main result that commitments cannot be grown is stated as
Theorem~\ref{thm:imp-c}. This is supplemented with an alternative
version of the claim, which applies if the initial functionality
enables committing to quantum bits.

\emph{String Commitments.}|\,A (classical) string commitment of length
$\ell$ is a functionality that takes a bit string $x\in \{0,1\}^\ell$
from the sender and outputs the message \texttt{committed} to the
receiver. Later, on input \texttt{open} from the sender, the
functionality sends $x$ to the receiver.

In the following, we consider implementations of this task by quantum
protocols between two parties, Alice (who holds system~$A$) and
Bob~($B$). They have access to a noiseless quantum and a noiseless
classical channel, as well as to an additional resource, $C$ (to be
specified later). In any round of the protocol, the parties may perform an
arbitrary quantum operation on the system in their possession
conditioned on the available classical information~\footnote{This assumption 
is not justified in the relativistic setting considered in~\cite{Kent11}.} \,|\, this includes
generating the input for the available communication interfaces.  The
use of the quantum channel then corresponds to a party transferring a
part of her system to the other party. The classical channel measures
the input in a canonical basis and sends the outcome to the
receiver. We assume that the total number of rounds of the protocol is
bounded by some finite number. By padding the protocol with empty
rounds, this corresponds to the assumption that the number of rounds
is equal in every execution. 

A string commitment scheme over strings of length $\ell$ generally
consists of two phases. In the first, the \emph{commit phase}, the
sender commits to an $\ell$-bit string~$x$. Later, in the
\emph{opening phase} the sender reveals $x$ to the receiver. The total system (consisting of the
subsystems controlled by Alice and Bob) is assumed to be in a pure
state initially. By introducing an additional space the quantum
operations of both parties can be purified, i.e., we can assume that
the parties apply, conditioned on the information shared over the
classical channel, isometries to their systems. Thus, we will assume
in the following that the state at the end of the commit phase
conditioned on all the classical communication is pure. %This allows us
%to construct attacks on protocols using an approach which is similar
%to what has been used for the impossibility proofs
%in~\cite{Mayers97,LoChau97,Lo97}.

\emph{Security Definitions.}|\,Our main technical contribution will be a quantitative statement on
the impossibility of growing string commitments. To formulate this
statement, we introduce two definitions that capture the cheating
probability of Alice and the information gain of Bob, respectively. We
emphasize that the properties required in these definitions are only
necessary (we therefore call the definitions ``weak''), but would not
be sufficient for the security of a protocol~\footnote{In particular,
  one would have to consider arbitrary malicious strategies of
  dishonest parties to prove the security of a protocol.}. Since we
are interested in the impossibility of certain protocols, this only
strengthens our results.

Using a commitment protocol, a (quantum) Alice can always commit to a
superposition of strings \cite{Mayers97, DMS00} as follows: she
prepares a state $\frac{1}{\sqrt{|\mX|}}\sum_{x\in
  \mX}\ket{x}_X\kron\ket{x}_{X'}$, where $\mX$ is a subset of the
$\ell$-bit strings. Then she honestly executes the commit protocol
with the first half of this state as input and keeps the system
$X'$. We denote the resulting joint state of Alice, Bob and the
resource system by $\rho^{\mX}_{A'BC}$, where $A'$ stands for
$XX'A$. Later, Alice can measure $X'$ and execute the opening phase of
the protocol with the resulting string $x$. Thus, even for a perfectly
binding commitment scheme, we cannot require that there is a fixed
value $x$ Alice is committed to after the commit phase. Rather, we can
only demand that $\sum_{x \in \{0,1\}^n}p_{x}\leq 1$ where $p_x$ is
the probability that Alice successfully reveals some $x$ in the
opening phase.

In order to quantify the degree of bindingness of a protocol, we
consider the following attack by Alice. First, she commits to a
superposition of strings from a set $\mX_0\subseteq \{0,1\}^\ell$ as
before.  Then, she tries to map (by a local transformation $\mE_A$ on
her system) the resulting state $\rho^{\mX_0}_{A'BC}$ to
$\rho^{\mX_1}_{A'BC}$, corresponding to the commitment to a set $\mX_1
\subseteq \{0,1\}^\ell$ which is disjoint from $\mX_0$.  Such an
attack is successful with probability at least $\Delta$ if the
protocol cannot detect the transformation with probability more than
$1-\Delta$. Using the \emph{trace distance},
$D(\rho,\tau):=\frac{1}{2} \norm{\rho - \tau}{1}$, this can be turned
into a necessary condition for security, formulated in terms of the
closeness of the transformed state, $(\mE_{A'} \kron
\idBC)(\rho^{\mX_0}_{A'BC})$, to the target state
$\rho^{\mX_1}_{A'BC}$.
 
\begin{definition2}[Weakly $\Delta$-binding]\label{def:binding}
 We call a commitment scheme \emph{weakly $\Delta$-binding} if 
\begin{align*}
 \min_{\mX_0,\mX_1} \min_{\mE_{A'}}\dis\left((\mE_{A'} \kron \idBC)(\rho^{\mX_0}_{A'BC}),\rho^{\mX_1}_{A'BC}\right) \geq 1 - \Delta\;,
\end{align*}
where $\mX_0$ and $\mX_1$ are disjoint sets of strings from $\{0,1\}^\ell$ and $\mE_{A'}$ is a completely positive trace preserving map 
acting on Alice's system.
\end{definition2}

To define the hiding property, we consider the joint state
$\rhoAB^{x}$ of Alice's and Bob's systems that results from an
execution of the protocol where both parties are honest and Alice
commits to $x$.  For a commitment scheme to be $\eps$-hiding, we
require that $\dis(\rhoB^{x},\rhoB^{x'})\leq \eps$ for any
$x,x'$. This immediately implies the following (necessary) security
condition.

\begin{definition2}[Weakly $\eps$-hiding]
A bit commitment protocol is \emph{weakly $\eps$-hiding} for uniform $X$ if the marginal state $\rho_{XB}$ after the commit phase is $\eps$-close to a 
state where $X$ is uniform with respect to $B$, i.e.,
\begin{align}\label{eq:hiding}
\min_{\sigmaB}\dis(\rho_{XB},\frac{1}{|X|}\id_{X}\kron \sigmaB)\leq \eps\;.
\end{align}
\end{definition2}

\emph{Smooth Entropies.}|\,Our proof is based on the insight that
every conceivable protocol that aims to extend bit commitment allows
for an attack, which can be established using known results on privacy
amplification and the smooth entropy formalism. (Privacy amplification
has also been used in \cite{BCHLW06} to construct attacks on
commitment schemes.) The detailed proofs of the technical statements
can be found in~\footnote{See EPAPS Document No.[number will be
  inserted by publisher].}.

Let $\rho_{XB}=\sum_{x}P(x)\ket{x}\bra{x}\kron \rhoB^x$ be a classical-quantum (CQ)
state. Then the min-entropy of $X$ conditioned on $B$, denoted $\chmin{X}{B}{\rho}$, corresponds to the
negative logarithm of the probability of guessing $X$ correctly from a
quantum memory $B$~\cite{KRS09}. The smooth min-entropy of a state is defined as $\chmineps{\eps}{X}{B}{\rho} := \max_{\rhot}\chmin{X}{B}{\rhot}$, where the optimization is over all (sub-normalized) states $\eps$-close to $\rho_{XB}$ in terms of the purified distance, which corresponds to the minimum trace distance between their purifications. The purified distance between two states, $\rho$ and $\rhot$, is upper bounded by $\sqrt{2\dis(\rho,\rhot)}$ \cite{TCR10}. %(see \cite{TCR10} for details). 

The \emph{leftover hash lemma} against quantum side information~\cite{Renner2005} (see also~\cite{Tom2010}) asserts that the smooth min-entropy of $\chmineps{\eps}{X}{B}{\rho}$ characterizes the amount of uniform randomness that can be extracted from $X$ with respect to the quantum side information $B$. A consequence of this is the following fact:
%that we use in our proof: 
for any CQ state $\rho_{XB}=\frac{1}{2^\ell}\sum_{x \in \{0,1\}^\ell}\ket{x}\bra{x}\kron \rhoB^x$ there exists a function $f: \{0,1\}^\ell\rightarrow \{0,1\}$ such that
\begin{align}\label{eq:leftover}
  \dis(\rho_B^{f,\mX_0},\rho_B^{f,\mX_1})\leq 2 \epsilon  + \sqrt{2^{1-\chmineps{\eps}{X}{B}{\rho}}}\;,
\end{align}
where $\rho_B^{f,\mX_z}=\frac{1}{|f^{-1}(z)|}\sum_{x \in f^{-1}(z)}\rho_B^{x}$.

In order to derive bounds on the conditional min-entropy when the
conditioning system is manipulated, we use the following
data-processing inequalities. Let $\rho_{XBC}$ be a CQ state, where
$C$ is an additional quantum register with dimension $\abs{C}$. Then,
the min-entropy $\chmineps{\eps}{X}{BC}{\rho}$ cannot increase by more
than $\log \abs{C}$ when a projective measurement $C \!\to\! Z$ is
applied,
\begin{align}\label{ineq:measurement}
  	\chmineps{\eps}{X}{BC}{\rho} \geq \chmineps{\eps}{X}{BZ}{\rho} - \log \abs{C} \;.
\end{align}
Moreover, if the classical register $Z$ is discarded, we have
\begin{align}\label{ineq:classical-chain}
  \chmineps{\eps}{X}{BZ}{\rho}\geq \chmineps{\eps}{X}{B}{\rho} - \log \abs{Z}\;.
\end{align}

The following fact, also used in the proofs of \cite{Mayers97,LoChau97,Lo97}, is an essential building block of our impossibility proofs: let $\phi_{AB}^0$ and $\phi_{AB}^1$ be two pure states corresponding to the joint state of Alice and Bob when committing to '0' and '1', respectively. If the marginal state of $\phi_{AB}^0$ and $\phi_{AB}^1$ on Bob's system is (almost) the same, then there exists a unitary $U_A$ on Alice system that (approximately) transforms $\phi_{AB}^0$ into $\phi_{AB}^1$, i.e., $(U_A\kron\idB)\ket{\phi_{AB}^0}\approx \ket{\phi_{AB}^1}$. This reasoning can be generalized to joint states $\rho_{YAYB}^b$ that are pure conditioned on all the classical information $Y$ available to both Alice and Bob as follows. If
$\dis(\rho_{YB}^0,\rho_{YB}^1)\leq\eps$, then there exists a unitary $U_{YA}$ such that 
\begin{align}\label{eq:classical-attack}
\dis\big(U_{YA}\,\rho_{YAYB}^0\, U_{YA}^\dagger,\rho_{YAYB}^1\big)\leq \sqrt{2 \eps}\;,
\end{align}
where we omitted the identity operator on $YB$.

\emph{Main Result}|\,One can trivially implement a string commitment
of length $n$ from $n$ bit commitments. Furthermore, it is easy to see
that, using a resource which allows the parties to commit to $n$
qubits, one can implement $n$ individual commitments to two bits each
using superdense coding \cite{BW92}, and, therefore, also a string
commitment of length $2n$. Our main result essentially states that
these two trivial implementations are essentially optimal.

More precisely, we first consider implementations of string
commitments based on a functionality that enables $n$ perfect
(classical) bit commitments. We show that the length of the
implemented string commitment is approximately upper bounded by $n$ if
this is required to be highly binding and hiding.

\begin{theorem}\label{thm:imp-c}
  Every quantum protocol which uses $n_A$ bit commitments from Alice
  to Bob and $n_B$ bit commitments from Bob to Alice with $n=n_A+n_B$
  as a resource and implements an $\eps$-hiding and $\Delta$-binding
  string commitment of length $\ell$ must satisfy
\begin{align*}
 \ell \leq n-2\log\left(\frac{(1-\Delta)^2}{4}-\sqrt{2\eps}\right)-1\;.
\end{align*}
In particular, if $\Delta=\eps\leq 0.01$, then $\ell < n+6$.
\end{theorem}

\begin{proof}

\cancel{
Let  $\rho_{XABC}=\sum_x \frac{1}{2^\ell}\vecstate{x}\kron\rhoABC^x$ be the state at the end of the commit phase of an $\eps$-hiding commitment protocol when the input $X$ of Alice is uniformly distributed. Since $\rho_{XB}$ must be $\eps$-close to uniform and the purified distance between $\rho$ and $\rho'$ is upper bounded by $\sqrt{2\dis(\rho,\rho')}$ \cite{TCR10}, the definition of the smooth min-entropy implies, with $\epst:=\sqrt{2\eps}$, that 
\begin{align}\label{ineq:hiding}
\chmineps{\epst}{X}{B}{\rho}\geq \log |X|=\ell.
\end{align}
}

In the following, we construct an attack by Alice on a modified
protocol that does not use the resource bit commitments and is not necessarily hiding. 
In this protocol we make Bob more powerful in the sense that he can
simulate the original protocol locally. Thus, any successful attack of
Alice against the modified protocol implies a successful attack
against the original protocol. %Since we only make use of the modified protocol to construct an attack against Bob, the modified protocol does not have to be hiding.

In the modified protocol, Alice, instead of using the resource bit
commitments, measures the bits to be committed, stores a copy and
sends them to Bob, who stores them in a classical register,
$C_A$. When one of these commitments is opened, he moves the
corresponding bit to his register $B$. Bob simulates the action of his
commitments locally as follows: instead of measuring a register, $Y$,
and sending the outcome to the commitment functionality, he applies
the isometry $U:\ket{y}_Y\mapsto \ket{yy}_{YY'}$ purifying the
measurement of the committed bit and stores $Y'$ in another register,
$C_B$. When Bob has to open the commitment, he measures $Y'$ and sends
the outcome to Alice over the classical channel. Furthermore, the state conditioned on the
classical communication is again pure.

Let $\rho_{XABC}=\frac{1}{2^\ell}\sum_x \vecstate{x}\kron \rhoABC^x$, where $C$ stands for $C_AC_B$, be the
state resulting from the execution of the modified protocol when the
input $X$ of Alice is uniformly
distributed. Its marginal state, $\rho_{XAB}$, is the corresponding state at the end of the commit phase of the original commitment protocol. The state $\rho_{XB}$ must be weakly $\eps$-hiding. %The purified distance between $\rho$ and $\rho'$ is upper bounded by $\sqrt{2\dis(\rho,\rho')}$ \cite{TCR10}. 
Thus, by the definition of the smooth min-entropy and setting $\epst:=\sqrt{2\eps}$, we get 
\begin{align}\label{ineq:hiding}
\chmineps{\epst}{X}{B}{\rho}\geq \log |X|=\ell.
\end{align}
Therefore, inequalities~\eqref{ineq:measurement} and~\eqref{ineq:classical-chain}  imply that
\begin{align}\label{ineq:ent-lower-bound}
\chmineps{\epst}{X}{BC_AC_B}{\rho}\geq \chmineps{\epst}{X}{B}{\rho}-n\geq\ell-n\;.
\end{align}
From \eqref{eq:leftover} we know that there exists a function~$f$ such that $\dis(\rho^{\mX_0}_{BC},\rho^{\mX_1}_{BC})\leq 2\delta $, where $\delta:=\epst  + \frac{1}{2} \sqrt{2^{1-\chmineps{\epst}{X}{BC}{\rho} }}$ and $\rho_{BC}^{\mX_z}=\frac{1}{|f^{-1}(z)|}\sum_{x\in f^{-1}(z)}\rho_{BC}^{x}$. In order to construct a concrete attack, let Alice choose a bit $z$ and commit to a uniform superposition of all strings $x$ with $f(x)=z$. Then the resulting joint state $\rho_{A'BC}^{\mX_z}$ at the end of the commit phase is pure conditioned an all the shared classical information. According to \eqref{eq:classical-attack} there exists, therefore, a unitary $U_{A'}$ on Alice's system that transforms $\rho_{A'BC}^{\mX_z}$ into a state which is $2\sqrt{\delta}$-close to  $\rho_{A'BC}^{\mX_{1-z}}$ in terms of the trace distance. The definition of weakly $\Delta$-binding implies that $1-\Delta \leq 2\sqrt{\delta}$ and, together with~\eqref{ineq:ent-lower-bound}, the statement follows.
\end{proof}

Next, we consider protocols which use a quantum commitment
functionality that allows the parties to commit to (and later reveal)
$n$ qubit states. By slightly modifying the proof of the theorem, we
show that there cannot exist a protocol that uses such a resource and
implements a string commitment of length larger than $2n$. We consider again a modified protocol, where
Bob simulates the resource system as follows: Alice, instead of using
the resource, sends the committed qubits to Bob, and Bob keeps all the
qubits that he would send to the commitment functionality in the
original protocol in a register, $C$. Let
$\rho_{XABC}$ be the joint state after the execution of the commit phase 
when Alice's input $X$ is uniformly distributed. We have
$\chmineps{\epst}{X}{B}{\rho}\geq \log |X|=\ell$ as in
\eqref{ineq:hiding}. Inequalities \eqref{ineq:measurement} and
\eqref{ineq:classical-chain} together imply that conditioning on an
additional quantum system $C$ cannot decrease the smooth min-entropy
by more than $2 \log \abs{C}$. Thus, we have
\begin{align}\label{ineq:ent-lower-bound2}
\chmineps{\epst}{X}{BC}{\rho}\geq \chmineps{\epst}{X}{B}{\rho}-2\log |C|=\ell-2n\;.
\end{align}  
Now we proceed as in the proof of the main theorem to get
\begin{align}\label{ineq:quantum-bound}
 \ell \leq 2n-2\log\left(\frac{(1-\Delta)^2}{4}-\sqrt{2\eps}\right)-1\;.
\end{align}
Note that the same reasoning applies to any resource which can be simulated by Bob such that the resulting state at the end of the commit phase is pure conditioned on all the classical communication and the simulated resource uses an additional memory of size at most $\log \abs{C}$. Thus, inequality~\eqref{ineq:quantum-bound} holds for arbitrary such resources with $\log |C|\leq n$.

\emph{Conclusions}|\,We proved that it is impossible to use a small
number of bit commitments as a resource to implement a larger string
commitment that is both arbitrarily binding and hiding.  This is in
stark contrast to corresponding positive results for other
cryptographic primitives, such as \textit{quantum key distribution} or
\emph{coin flipping}, where the resource of interest, once available
in finite number, can be enlarged ad infinitum.

The techniques we use to show our impossibility results can be applied to prove more general results on the possibility and efficiency of
two-party cryptography. In particular, they can be used to prove bounds on the efficiency of implementations of string commitments from
oblivious transfer and, more generally, from resources that distribute trusted correlations to the parties. Moreover, the
impossibility results on implementations of oblivious transfer presented in \cite{WW10} can be improved using these techniques.

\emph{Acknowledgments.}|\,We thank Fr\'ed\'eric Dupuis and J\"urg Wullschleger for helpful and inspiring discussions. We acknowledge support from the Swiss National Science Foundation (grant no.\
200020-135048), the European Research Council (grant no.\ 258932), and
an ETHIIRA grant of ETH's research commission.

\bibliographystyle{apsrev4-1}
%\bibliography{all}

\begin{widetext}
\section{appendix}
Section~A contains general definitions and technical lemmas
related to distance measures and the smooth entropy calculus, as
needed for our work. In Section~B we present the full
proofs of our main results.

\subsection{A. Preliminaries}\label{sec:pr}
We restrict our attention to finite-dimensional Hilbert spaces $\h$. We use $\posops{\h}$ to denote the set of positive semi-definite operators on $\h$. We define the set of normalized quantum states by $\normstates{\h} := \{ \rho \in \posops{\h} : \tr\,\rho = 1 \}$ and the set of sub-normalized states by $\subnormstates{\h} := \{ \rho \in
\posops{\h} : 0 < \tr\,\rho \leq 1 \}$.  Given a state $\rhoAB \in \normstates{\hA\kron\hB}$ we denote by $\rhoA$ and $\rhoB$ its marginal states $\rhoA=\ptrace{B}{\rhoAB}$ and $\rhoB=\ptrace{A}{\rhoAB}$. We define the fidelity between two states $\rho,\tau \in \normstates{\rho}$ as $F(\rho,\tau)=\norm{\sqrt{\rho}\sqrt{\tau}}{1}$. For $\rho,\tau \in \normstates{\hA}$, we define the 
\emph{trace distance} between $\rho$ and $\tau$ as
\begin{align*}
 \dis(\rho,\tau):=\frac{1}{2} \norm{\rho - \tau}{1}.%+\frac{1}{2}|\tr\,\rho-\tr\, \tau|
\end{align*}
For $b \in \{0,1\}$, let $\rho_{XB}^b=\sum_x\vecstate{x}\kron \rho_B^{x,b}$ be classical-quantum (CQ) states. Then we have (see \cite{Renner2005} for a proof)
\begin{align}\label{equ:distance-mix}
\norm{\rho_{XB}^0-\rho_{XB}^1}{1}=\sum_{x \in \mX}\norm{\rho_B^{x,0}-\rho_B^{x,1}}{1}.                                                                                                                                                         
\end{align}                                                                                                                                                             
% 
%
%  DISTANCE FROM UNIFORM
%
%
\begin{definition}
  \label{def:dist-from-uniform}
  For $\rhoAB \in \normstates{\hAB}$ we define the \emph{distance
    from uniform of A conditioned on B} as
  \begin{align}
    \label{eqn:dist-from-uniform}
    \duni{A}{B}{\rho} := \min_{\sigmaB}\
     \dis(\rhoAB, \omega_{A} \kron \sigmaB)\,,
  \end{align}
  where $\omega_{A} := \idA/\dim \hA$ and the minimum is taken over all $\sigmaB \in \normstates{\hB}$.
\end{definition}

\begin{lemma}\label{lem:distance}
 Let $\rho_{XB}=\sum_{x \in \{0,1\}}\frac{1}{2}\ket{x}\bra{x}\kron \rhoB^x$ be a CQ state and $\duni{X}{B}{\rho}\leq \eps$. Then
\[\dis(\rhoB^0,\rhoB^1)\leq 2\eps.\]
\end{lemma}
\begin{proof}
$\dis(\rho_{XB}, \omega_{A} \kron \sigmaB)\leq \eps$ implies
\[\norm{\rhoB^0-\rhoB^1}{1}\leq\norm{\rhoB^0-\sigmaB}{1}+\norm{\rhoB^1-\sigmaB}{1}\leq 4\eps\]
where we used (\ref{equ:distance-mix}) and, therefore, we have $\dis(\rhoB^0,\rhoB^1)\leq 2\eps$.
\end{proof}

Furthermore, we will make use of the following well-known technical lemma which is also used in \cite{Mayers97,LoChau97,Lo97}.
\begin{lemma}\label{lem:uhlmann}
Let $\ket{\psiAB^0}$ and $\ket{\psiAB^1}$ be states with $\dis(\rhoB^0,\rhoB^1)\leq \eps$ where $\rhoB^x=\ptr{A}{\ket{\psiAB^x}\bra{\psiAB^x}}$. Then there exists a unitary $U_{A}$ such that 
\[\dis(\ket{\phiAB^1}\bra{\phiAB^1},\ket{\psiAB^1}\bra{\psiAB^1})\leq \sqrt{2\eps}\]
with $\phiAB^1=(U_{A}\kron \idB)\ket{\psiAB^0}$.
\end{lemma}
\begin{proof}
$\dis(\rhoB^0,\rhoB^1)\leq \eps$ implies $F(\rhoB^0,\rhoB^1)\geq 1-\eps$. From Uhlmann's theorem we know that there exists a unitary $U_A$ such that
$F(\ket{\phiAB^1}\bra{\phiAB^1},\ket{\psiAB^1}\bra{\psiAB^1})\geq 1-\eps$ where $\ket{\phiAB^1}=(U_{A}\kron \idB)\ket{\psiAB^0}$. Since $\dis(\rho,\tau)\leq \sqrt{1-F(\rho,\tau)^2}$ for any $\rho,\tau \in \normstates{\h}$ \cite{FG99}, we have $\sqrt{1-\dis(\ket{\phiAB^1}\bra{\phiAB^1},\ket{\psiAB^1}\bra{\psiAB^1})^2}\geq 1-\eps$. Hence,
\begin{align*}
\dis(\ket{\phiAB^1}\bra{\phiAB^1},\ket{\psiAB^1}\bra{\psiAB^1})\leq \sqrt{1-(1-\eps)^2}\leq \sqrt{2\eps}
\end{align*}

\end{proof}
Lemma \ref{lem:uhlmann} can be generalized to states which are pure conditioned on all classical information available to both $A$ and $B$ in the following way.
\begin{lemma}\label{lem:classical-attack}
For $b \in \{0,1\}$, let 
\[\rho_{XX'AB}^b=\sum_xP_b(x)\vecstate{x}_{X}\kron \vecstate{x}_{X'}\kron \vecstate{\psiAB^{x,b}}\]
with $\dis(\rho_{X'B}^0,\rho_{X'B}^1)\leq \eps$. Then there exists a unitary $U_{AX}$ such that 
\[\dis(\rho_{XX'AB}'^1,\rho_{XX'AB}^1)\leq 2 \eps \]
where $\rho_{XX'AB}'^1=(U_{XA}\kron \id_{X'B})\rho_{XX'AB}^0(U_{XA}\kron \id_{X'B})^\dagger$.
\end{lemma}
\begin{proof}
Define $\ket{\psi_{XX'X''AB}^b}:=\sum_x\sqrt{P_b(x)}\ket{x}_X\kron\ket{x}_{X'}\kron\ket{x}_{X''}\kron\ket{\psiAB^{x,b}}$ and let
\[\rho_{X'X''B}^b=\ptrace{XA}{\vecstate{\psi_{XX'X''AB}^b}}.\]%\rhob_{ABC_AC_B}}=\rho_{BC_B}^b.\]
Then
\[\dis(\rho_{X'X''B}^0,\rho_{X'X''B}^1)=\dis(\rho_{X'B}^0,\rho_{X'B}^1)\leq \eps\]
Thus, Lemma \ref{lem:uhlmann} implies the existence of a unitary $U_{AX}$ such that 
\[\dis(\vecstate{\phi_{XX'X''AB}^1},\vecstate{\psi_{XX'X''AB}^1})\leq \sqrt{2\eps}\]
with $\ket{\phi_{XX'X''AB}^1}=(U_{AX}\kron \id_{X'X''B})\ket{\psi_{XX'X''AB}^0}$. The statement then follows from the fact that taking the partial trace over $X''$ cannot increase the trace distance and commutes with the unitary $U_{AX}$ as follows. Let  $\rho_{XX'AB}'^1=(U_{XA}\kron \id_{X'B})\rho_{XX'AB}^0(U_{XA}\kron \id_{X'B})^\dagger$. Then
\begin{align*}
\dis((U_{XA}\kron \id_{X'B})&\rho_{XX'AB}^0(U_{XA}\kron \id_{X'B})^\dagger,\rho_{XX'AB}^1)\\
&=\dis((U_{XA}\kron \id_{X'B})\ptrace{X''}{\rho_{XX'X''AB}^0}(U_{XA}\kron \id_{X'B})^\dagger,\ptrace{X''}{\rho_{XX'X''AB}^1})\\
&=\dis(\ptrace{X''}{(U_{XA}\kron \id_{X'X''B})\rho_{XX'X''AB}^0(U_{XA}\kron \id_{X'X''B})^\dagger},\ptrace{X''}{\rho_{XX'X''AB}^1})\\
&\leq\dis((U_{XA}\kron \id_{X'X''B})\rho_{XX'X''AB}^0(U_{XA}\kron \id_{X'X''B})^\dagger,\rho_{XX'X''AB}^1)\\
&\leq \sqrt{2\eps}
\end{align*}
\end{proof}

%\paragraph{(Smooth) Min-Entropy}

We define the non-smooth min-entropy as follows. 

\begin{definition}[Min-Entropy]
\begin{align*}
\chmin{A}{B}{\rho} := \max_{\sigmaB \in \normstates{\hB}}\ \sup \big\{
  \lambda \in \mathbb{R} :  2^{-\lambda}\, \idA \kron \sigmaB \geq
  \rhoAB \big\} \, .
\end{align*}
\end{definition}

Then we define the smooth version of the min-entropy of a state $\rho$ as an optimization of the non-smooth entropy over a set of states that are close to $\rho$. As a distance measure between two states we use the purified distance, which corresponds to the minimum trace distance between purifications of these states \cite{TCR10}. 
\begin{definition}[Purified Distance]
  For $\rho, \tau \in \subnormstates{\h}$, we define the
  \emph{purified distance} between $\rho$ and $\tau$ as
\begin{equation*}
P(\rho, \tau) := \sqrt{1 - \bar{F}(\rho, \tau)^2}
\end{equation*}
where the generalized fidelity $\bar{F}$ is defined as $\bar{F}(\rho, \tau)=F(\rho,\tau)+\sqrt{(1-\tr\,\rho)(1-\tr\,\tau)}$. Note that $\bar{F}(\rho, \tau)=F(\rho, \tau)$ if at least one of the states is normalized.
\end{definition}

Let $\eps \geq 0$ and $\rho \in \subnormstates{\h}$ with
  $\sqrt{\tr\,\rho} > \eps$. Then, we define
  an \emph{$\eps$-ball} in $\h$ around $\rho$ as
  \begin{equation*}
    \label{eqn:def-eps-ball}
    \epsball{\eps}{\h ; \rho} := \{ \tau \in \subnormstates{\h} :
    P(\tau, \rho) \leq \eps \} \, .
  \end{equation*}

The smoothed version of the min-entropy is defined as follows.

\begin{definition}[Smooth Min-Entropy]
  \label{def:mineps}
  Let $\eps \geq 0$ and $\rhoAB \in \subnormstates{\hAB}$,
  then the \emph{$\eps$-smooth min-entropy} of A conditioned on B of
  $\rhoAB$ is defined as
  \begin{equation*}
    \chmineps{\eps}{A}{B}{\rho} := \!\! \max_{\rhotAB \in
      \epsball{\eps}{\rhoAB}} \! \chmin{A}{B}{\rhot}\, .
  \end{equation*}
\end{definition}

A family $\mF$ of functions from $\mX$ to $\mZ$ is called weakly two-universal \cite{CW79} if for any pair of distinct inputs $x$ and $x'$ the probability of a collision $f(x)=f(x')$ is at most $1/|\mZ|$ if $f$ is chosen at random from $\mF$. The following lemma \cite{Renner2005} (see also \cite{Tom2010}) shows that weak two-universal hash functions are strong extractors against quantum side information, i.e., the output of the function is uniform with respect to the side information and the choice of the function. 
\begin{lemma}[Leftover Hash Lemma]\label{lem:leftover-hash} 
Let $\mF$ be a family of weak two-universal hash functions from $\mX$ to $\{0,1\}$. Let $\rho_{XB}=\sum_{x}P(x)\ket{x}\bra{x}\kron \rhoB^x$ be a CQ state and $\rho_{FZB}=\frac{1}{|\mF|}\sum_f\sum_z \ket{f}\bra{f}\kron\ket{z}\bra{z}\kron\rho_B^{f,z}$ with $z \in \{0,1\}$ and  $\rho_B^{f,z}=\sum_{x \in f^{-1}(z)}P(x)\rho_B^{x}$. Then
\begin{align*} 
  \duni{Z}{BF}{\rho} \leq \epsilon  + \frac{1}{2} \sqrt{2^{1-\chmineps{\eps}{X}{B}{\rho} }}.
\end{align*}
\end{lemma}

%As a consequence of Lemmas \ref{lem:distance} and \ref{lem:leftover} we get the following 
\begin{lemma}\label{lem:leftover} 
%Let $\mF$ be a family of weak two-universal hash functions from $\{0,1\}^\ell$ to $\{0,1\}$ such that $|\{x\in\{0,1\}^\ell:~f(x)=0\}|=2^{\ell-1}$. 
Let $\rho_{XB}=\frac{1}{2^\ell}\sum_{x \in \{0,1\}^\ell}\ket{x}\bra{x}\kron \rhoB^x$ be a CQ state. % and $\rho_{FZE}=\frac{1}{|\mF|}\sum_f\sum_z \ket{f}\bra{f}\kron\ket{z}\bra{z}\kron\rho_E^{f,z}$ with $z \in \{0,1\}$ and  $\rho_E^{f,z}=\sum_{x:~f(x)=z}\rho_E^{x}$. 
Then there exists a function $f: \{0,1\}^\ell\rightarrow \{0,1\}$ in $\mF$ such that
\begin{align*}
  \dis(\rho_B^{f,0},\rho_B^{f,1})\leq 2\left( \epsilon  + \frac{1}{2} \sqrt{2^{1-\chmineps{\eps}{X}{B}{\rho}}}\right),
\end{align*}
where $\rho_B^{f,z}=\frac{1}{|f^{-1}(z)|}\sum_{x \in f^{-1}(z)}\rho_B^{x}$.
\end{lemma}
\begin{proof}
 Let $\mF$ be a family of two-universal hash functions $f:\{0,1\}^\ell \rightarrow \{0,1\}$ such that every $f$ is balanced, i.e., $|\{x\in\{0,1\}^\ell:~f(x)=0\}|=2^{\ell-1}$. From Lemma \ref{lem:leftover-hash} we know that 
\[\duni{Z}{BCF}{\rho} \leq \delta\]%\epsilon  + \frac{1}{2} \sqrt{2^{1-\chmineps{\eps}{Z}{E}{\rho} }}.\]
where $\delta:=\eps  + \frac{1}{2} \sqrt{2^{1-\chmineps{\eps}{X}{B}{\rho} }}$ and $Z:=f(X)$. Thus, there must exist a function $f \in \mF$ such that $\duni{Z}{B}{\rho^{[f]}} \leq~\delta$. 
For $z \in \{0,1\}$ let 
\begin{align*}
 \rho_{B}^{f,z}=\frac{1}{2^{\ell-1}}\sum_{x\in f^{-1}(z)}\rho_{B}^{x}.
\end{align*}
From Lemma~\ref{lem:distance} we then have $\dis(\rho^{f,0}_{BC},\rho^{f,1}_{BC})\leq 2\delta $. 
\end{proof}

The following lemma shows that the conditional min-entropy $\chmineps{\eps}{A}{B}{\rho}$ can decrease by at most $\log |Z |$ when conditioning on an additional classical system $Z$.

\begin{lemma}\label{lem:chain-rule-c}
Let $\eps>0$ and let $\rhoABZ$ be a tripartite state that is classical on $Z$ with respect to some orthonormal basis $\{\ket{z}\}_z$. Then
\[\chmineps{\eps}{A}{BZ}{\rho}\geq \chmineps{\eps}{A}{B}{\rho}-\log |Z|.\]
\end{lemma}
\begin{proof}
  Let $\rhotAB$ be the state that optimizes the min-entropy
  $\chmineps{\eps}{A}{B}{\rho} = \chmin{A}{B}{\rhot}$. Then, there exists an
  extension $\idx{\rhot}{ABZ}$ of $\rhotAB$ that is $\eps$-close to
  $\idx{\rho}{ABZ}$ and classical on Z. See~\cite{TCR10}, where it is shown that there always exists an $\eps$-close extension and that the purified distance can only decrease under a measurement in the Z basis. Let $\idx{\rhot}{ABZ} =
  \sum_z \rhotAB^z \kron \proj{z}{z}$ so that $\rhotAB^z \leq \rhotAB$
  for all $z$. By the definition of the min-entropy, we have
  $$ \rhotAB^z \leq \rhotAB \leq 2^{-\chmineps{\eps}{A}{B}{}} \idA \kron
  \sigmaB $$ for the optimal $\sigmaB$. Hence, $$ \rhotABC = \sum
  \rhotAB^z \kron \proj{z}{z} \leq 2^{-\chmineps{\eps}{A}{B}{}} \idA \kron \sigmaB \kron \idZ \, .
$$ The lemma now follows from the definition of
  the min-entropy $\chmineps{\eps}{A}{BZ}{\rho}$, where
  $\idx{\rhot}{ABZ}$ and $\idx{\sigma}{BZ} = \sigmaB \kron \idZ /
  \abs{\textnormal{Z}}$ are candidates for the optimization.
\end{proof}

The following lemma shows that the min-entropy $\chmineps{\eps}{A}{BC}{\rho}$ cannot increase too much when a projective measurement is applied to system $C$.

 \begin{lemma}
  \label{lem:data-processing-bound}
  Let $\eps \geq 0$ and let $\rhoABC$ be a tri-partite state. Furthermore, 
  let $\cM$ be a projective measurement in the basis $\{
  \ket{z} \}_z$ on C and 
  $\idx{\rho}{ABZ} := \idx{\cI}{AB} \kron \cM (\rhoABC)$, where $\idx{\cI}{AB}$ is the identity operation on A and B. Then,
  \begin{align*}
  	\chmineps{\eps}{A}{BC}{\rho} \geq \chmineps{\eps}{A}{BZ}{\rho} - 
  	\log \abs{\textnormal{Z}} \, .
  \end{align*}
\end{lemma}

\begin{proof}
  Let $U: \keti{z}{C} \mapsto \keti{zz}{ZZ$'$}$ be the isometry purifying $\cM$
  in the sense that $\idx{\rho}{ABZ} = \ptr{Z$'$} (\idx{\rho}{ABZZ'})$, where 
  $\idx{\rho}{ABZZ'} := U \rhoABC U^\dagger$. Covariance under isometries of
  the smooth min-entropy implies
  \begin{align*}
  	\chmineps{\eps}{A}{BC}{\rho} = \chmineps{\eps}{A}{BZZ'}{\rho} \, .
  \end{align*}
  Moreover, for some states $\idx{\rhot}{ABZ}$ and $\idx{\sigmat}{BZ}$, we have
  \begin{align}
  	\chmineps{\eps}{A}{BZ}{\rho} &= \sup \big\{ \lambda \in \mathbb{R} : \idx{\rhot}{ABZ} 
  	\leq 2^{-\lambda}\, \idA \kron \idx{\sigmat}{BZ} \big\} \nonumber \\
  	&\leq \sup \big\{ \lambda \in \mathbb{R} : \idx{\rhot}{ABZZ'} \leq 2^{-\lambda}\,
  	 \abs{\textnormal{Z}} \, \idA \kron \idx{\sigmat}{BZZ'}  \big\} \label{eqn:rel} \\
  	&\leq \chmineps{\eps}{A}{BZZ'}{\rho} + \log \abs{\textnormal{Z}} \nonumber \, . 
  \end{align}
  Here, $\idx{\rhot}{ABZZ'}$ is an extension of $\idx{\rhot}{ABZ}$ that is $\eps$-close
  to $\idx{\rho}{ABZZ'}$ and satisfies 
  $\idx{\Pi}{ZZ'} \idx{\rhot}{ABZZ'} \idx{\Pi}{ZZ'} = \idx{\rhot}{ABZZ'}$, 
  where $\idx{\Pi}{ZZ'} := \sum_z \proji{zz}{zz}{ZZ'}$. The existence of such an extension can be deduced from the fact that projections can only decrease the purified distance~\cite{TCR10} and $\idx{\Pi}{ZZ'}$ commutes with $\idx{\rho}{ABZZ'}$. Furthermore, 
  $\idx{\sigmat}{BZZ'} := \idx{\Pi}{ZZ'} (\idx{\sigmat}{BZ} \kron \idi{Z'}) \idx{\Pi}{ZZ'}$.
  The last inequality follows since $\idx{\rhot}{ABZZ'}$ and $\idx{\sigmat}{BZZ'}$ are
  candidates for the optimization of the min-entropy. It remains to show the implication
  \begin{align}
  	\idx{\rhot}{ABZ} \leq 2^{-\lambda}\, \idA \kron \idx{\sigmat}{BZ} 
  	\implies \idx{\rhot}{ABZZ'} \leq 2^{-\lambda}\,
  	 \abs{\textnormal{Z}} \, \idA \kron \idx{\sigmat}{BZ} \kron \idi{Z'} 
	 \label{eqn:maxext}
  \end{align}
  which in turn implies~\eqref{eqn:rel}. However,~\eqref{eqn:maxext} follows from the 
  fact that, for any extension $\idx{X}{AB}$ of a positive operator $\idx{X}{A}$, it 
  holds that $\idx{X}{AB} \leq \abs{\textnormal{B}}\, \idx{X}{A} \kron 
  \idB$. Since $\idx{X}{AB}$ has a spectral decomposition with positive coefficients, it is sufficient to show this property for pure normalized states $\proji{\psi}{\psi}{AB}$. The general property then follows by taking the weighted sum on both sides of the inequality. Let $\tauA := \ptr{B}(\proji{\psi}{\psi}{AB})$ and $\idx{\Gamma}{AB} := (\tauA^{-\frac{1}{2}} \kron \idB) \proji{\psi}{\psi}{AB}\, (\tauA^{-\frac{1}{2}} \kron \idB)$, where the inverse is taken on the support of $\tauA$. Since $\idx{\Gamma}{AB}$ is of rank $1$, its maximum eigenvalue is $\tr (\idx{\Gamma}{AB}) = \textrm{rank} \{\tauA\} \leq \min \{ \abs{\textnormal{A}}, \abs{\textnormal{B}} \}$ and, thus, $\idx{\Gamma}{AB} \leq \abs{\textnormal{B}}\, \idAB$. Hence, by conjugation of both sides with $\tauA^{\frac{1}{2}}$ follows $\proji{\psi}{\psi}{AB} \leq \abs{\textnormal{B}}\, \tauA \kron \idB$.
  This concludes the proof.
\end{proof}

The following lemma, which shows that conditioning on an additional quantum system $C$ cannot decrease the conditional smooth min-entropy by more than $2 \log \abs{C}$, follows immediately from  Lemmas \ref{lem:chain-rule-c} and \ref{lem:data-processing-bound}  
\begin{lemma}\label{lem:chain-rule-q}
 $\chmineps{\eps}{A}{BC}{\rho}\geq \chmineps{\eps}{A}{B}{\rho}-2\log |C|.$
\end{lemma}

\subsection{B. Main Results}\label{sec:main}

\subsubsection{(Classical) Bit Commitment Resource}
\begin{theorem}\label{thm:class-imp}
Every quantum protocol which uses $n_A$ (classical) bit commitments
from Alice to Bob and $n_B$ (classical) bit commitments from Bob to
Alice with $n=n_A+n_B$ as a resource and implements an $\eps$-hiding
and $\Delta$-binding string commitment of length at most
\begin{align*}
 \ell \leq n-2\log\left(\frac{(1-\Delta)^2}{4}-\sqrt{2\eps}\right)-1.
\end{align*}
In particular, if $\Delta=\eps\leq 0.01$, then $\ell \leq n+6$.
\end{theorem}
\begin{proof}
Let $\vecstate{x}\kron\rhoABC^x$ be the state resulting from the execution of an $\eps$-hiding commitment protocol when the input of Alice is $x$. Then $\rho_{XABC}=\sum_x \frac{1}{2^\ell}\vecstate{x}\kron\rhoABC^x$ is the state resulting from an execution where the committed string $X$ is uniformly distributed. Let $\epst:=\sqrt{2\eps}$. Since $\rho_{XB}$ is $\eps$-close to uniform and $P(\rho,\rho')\leq \sqrt{2\dis(\rho,\rho')}$ \cite{TCR10}, the definition of the smooth min-entropy implies that
\[\chmineps{\epst}{X}{B}{\rho}\geq \log |X|=\ell.\]
In the following, we consider a modified protocol that does not use the resource bit commitments. In this modified protocol Alice, instead of using the resource bit commitments, measures the bits to be committed, stores a copy and sends them to Bob, who stores them in a classical register $C_A$. When one of these commitments is opened, he moves the coresponding bit to his register $B$. Bob simulates the action of his commitments locally as follows: instead of measuring a register, $Y$, and sending the outcome to the commitment functionality, he applies the isometry $U:\ket{y}_Y\mapsto \ket{yy}_{YY'}$ purifying the measurement of the committed bit and stores $Y'$ in register $C_B$. When Bob has to open the commitment, he measures $Y'$ and sends the outcome to Alice over the classical channel. Note that we make Bob more powerful in this modified protocol because he can simulate the original protocol locally. Thus, any successful attack of Alice against the modified protocol implies a successful attack against the original protocol. Since we only make use of the modified protocol to construct an attack against Bob, the modified protocol does not have to be hiding. Furthermore, the state conditioned on the classical communication is again pure. Let 
$\vecstate{x}\kron\rhob_{AB}^x$ be the state resulting from the execution of the modified protocol when the input of Alice is $x$.
Then $\rhob_{XAB}=\sum_x \frac{1}{2^\ell}\vecstate{x}\kron\rhob_{AB}^x$ is the state resulting from an execution where the committed string $X$ is uniformly distributed. 
From Lemma \ref{lem:leftover} we know that there exists a function $f:\{0,1\}^\ell\rightarrow \{0,1\}$ such that 
\[\dis(\rho^{\mX_0}_{BC},\rho^{\mX_1}_{BC})\leq 2\delta\]
where $\rho_{BC}^{\mX_z}=\frac{1}{2^{\ell-1}}\sum_{x\in f^{-1}(z)}\rho_{BC}^{x}$, $\delta:=\epst  + \frac{1}{2} \sqrt{2^{1-\chmineps{\epst}{X}{BC}{\rho} }}$ and $C$ stands for $C_AC_B$. Let $z \in \{0,1\}$  and let Alice prepare the state
\begin{align*}
 \frac{1}{\sqrt{2^{\ell-1}}}\sum_{x\in f^{-1}(z)}\ket{x}_X\kron\ket{x}_{X'}
\end{align*}
and honestly executes the commit protocol with the first half of this state as input. Let $\rho_{A'BC_AC_B}^{\mX_z}=\rho_{XX'ABC_AC_B}^{\mX_z}$ be the resulting joint state at the end of the commit phase. Then we have $\ptrace{A'}{\rho_{A'BC_AC_B}^{\mX_z}}=\rho_{BC_AC_B}^{\mX_z}$ and, therefore, Lemma \ref{lem:classical-attack} then implies that there exists unitary $U_A$ such that
\begin{align}\label{ineq:dist-upper-bound-2}
 \dis(\rhot_{A'BC_AC_B}^{\mX_z},\rho_{A'BC_AC_B}^{\mX_z})\leq 2\sqrt{\delta},
\end{align}
where $\rhot_{A'BC_AC_B}^{\mX_z}=(U_{A'}\kron \idB)\rho_{A'BC_AC_B}^{\mX_z}(U_{A'}\kron \idB)^\dagger$. 
Lemmas \ref{lem:chain-rule-c} and \ref{lem:data-processing-bound} imply that 
\begin{align}\label{ineq:entropy-bound}
\chmineps{\epst}{X}{BC_AC_B}{\rho}&\geq \chmineps{\epst}{X}{BC_B}{\rho}-n_A\nonumber\\
&\geq \chmineps{\epst}{X}{B}{\rho}-n\nonumber\\
&\geq \ell-n
\end{align}
Thus, we have 
\begin{align*} 
1-\Delta &\leq 2\sqrt{\delta} = 2\sqrt{\epst  + \frac{1}{2} \sqrt{2^{1-\chmineps{\epst}{X}{BC_AC_B}{\rho}}}}\\
&\leq 2\sqrt{\epst + \frac{1}{2} \sqrt{2^{1-\ell+n}}}\\
&\leq 2\sqrt{\sqrt{2\eps}  +  2^{-\frac{1}{2}(\ell-n+1)}}
\end{align*}
where we used the definition of weakly $\Delta$-binding and inequalities~\eqref{ineq:dist-upper-bound-2} and~\eqref{ineq:entropy-bound}.
\end{proof}

\subsubsection{Quantum Resource}

Next, we consider implementations of string commitments from a functionality which allows the players to commit to (and later reveal) $n$ qubit states. The following theorem shows that there cannot exist a protocol using such a resource which implements an arbitrarily hiding and binding string commitment of length larger than $2n$.
\begin{theorem}\label{thm:quant-imp}
Every quantum protocol which uses a resource, which allows the players to commit to (and later reveal) $n$ qubit states and implements an $\eps$-hiding and $\Delta$-binding string commitment of length $\ell$ must have
\begin{align}\label{ineq:quantum-bound2}
 \ell \leq 2n-2\log\left(\frac{(1-\Delta)^2}{4}-\sqrt{2\eps}\right)-1.
\end{align} 
In particular, if $\Delta=\eps\leq 0.01$, then $\ell \leq 2n+6$.
\end{theorem}

\begin{proof}
Let $\vecstate{x}\kron\rhoABC^x$ be the state resulting from the execution of an $\eps$-hiding commitment protocol when the input of Alice is $x$. Then $\rho_{XABC}=\sum_x \frac{1}{2^\ell}\vecstate{x}\kron\rhoABC^x$ is the state resulting from an execution where the committed string $X$ is uniformly distributed. Let $\epst:=\sqrt{2\eps}$. Since $\rho_{XB}$ is $\eps$-close to uniform and $P(\rho,\rho')\leq \sqrt{2\dis(\rho,\rho')}$ \cite{TCR10}, the definition of the smooth min-entropy implies that
\[\chmineps{\epst}{X}{B}{\rho}\geq \log |X|=\ell.\]
%where we used the fact that $P(\rho,\rho')\leq \sqrt{2\dis(\rho,\rho')}$ \cite{TCR10}.
From Lemma \ref{lem:chain-rule-q} we have 
\begin{align}\label{ineq:ent-lower-bound-a}
\chmineps{\epst}{X}{BC}{\rho}\geq \chmineps{\epst}{X}{B}{\rho}-2\log |C|.
\end{align}

From Lemma \ref{lem:leftover} we know that there exists a function $f:\{0,1\}^\ell\rightarrow \{0,1\}$ such that 
\[\dis(\rho^{\mX_0}_{BC},\rho^{\mX_1}_{BC})\leq 2\delta\]
where $\rho_{BC}^{\mX_z}=\frac{1}{2^{\ell-1}}\sum_{x\in f^{-1}(z)}\rho_{BC}^{x}$ and $\delta:=\epst  + \frac{1}{2} \sqrt{2^{1-\chmineps{\epst}{X}{BC}{\rho} }}$.
Let $z \in \{0,1\}$ and let Alice prepare the state
\begin{align*}
 \frac{1}{\sqrt{2^{\ell-1}}}\sum_{x\in f^{-1}(z)}\ket{x}_{X}\kron\ket{x}_{X'}
\end{align*}
and honestly execute the commit protocol with the first half of this state as input. Let $\rho_{A'BC}^{\mX_z}=\rho_{XX'ABC}^{\mX_z}$ be the resulting state. Then we have $\ptrace{A'}{\rho_{A'BC}^{\mX_z}}=\rho_{BC}^{\mX_z}$ and, therefore, Lemma \ref{lem:classical-attack} implies that there exists a unitary $U_{A'}$ such that 
\begin{align}\label{ineq:dist-upper-bound}
\dis(\rhot_{A'BC}^{\mX_{1-z}},\rho_{A'BC}^{\mX_{1-z}})\leq 2\sqrt{\delta}
\end{align}
where $\rhot_{A'BC}^{\mX_{1-z}}=(U_{A'}\kron \id_{BC})\rho_{A'BC}^{\mX_{z}}(U_{A}\kron \id_{BC})^\dagger$.
This implies that
\begin{align*} 
1-\Delta &\leq 2\sqrt{\delta} = 2\sqrt{\epst  + \frac{1}{2} \sqrt{2^{1-\chmineps{\epst}{X}{BC}{\rho}}}}\\
&\leq 2\sqrt{\epst + \frac{1}{2} \sqrt{2^{1-\ell+2n}}}\\
&\leq 2\sqrt{\sqrt{2\eps}  +  2^{-\frac{1}{2}(\ell-2n+1)}}\\
\end{align*}
where we used the definition of weakly $\Delta$-binding and inequalities~\eqref{ineq:ent-lower-bound-a} and~\eqref{ineq:dist-upper-bound}.
\end{proof}
Note that the proof of Theorem~\ref{thm:quant-imp} only uses the fact that the resource could be simulated by Bob such that the resulting state at the end of the commit phase is pure conditioned on all the classical communication and the simulated resource uses an additional memory of size at most $\log \abs{C}$. Thus, inequality \eqref{ineq:quantum-bound2} holds for arbitrary such resources with $\log |C|\leq n$. A simple example of such a resource would be a functionality which generates a tripartite state $\ket{\phi}_{ABC}$ and gives system~$A$ to Alice and $B$ to Bob.

\end{widetext}

\begin{thebibliography}{10}%
\makeatletter
\providecommand \@ifxundefined [1]{%
 \ifx #1\undefined \expandafter \@firstoftwo
 \else \expandafter \@secondoftwo
\fi
}%
\providecommand \@ifnum [1]{%
 \ifnum #1\expandafter \@firstoftwo
 \else \expandafter \@secondoftwo
\fi
}%
\providecommand \enquote [1]{``#1''}%
\providecommand \bibnamefont  [1]{#1}%
\providecommand \bibfnamefont [1]{#1}%
\providecommand \citenamefont [1]{#1}%
\providecommand\href[0]{\@sanitize\@href}%
\providecommand\@href[1]{\endgroup\@@startlink{#1}\endgroup\@@href}%
\providecommand\@@href[1]{#1\@@endlink}%
\providecommand \@sanitize [0]{\begingroup\catcode`\&12\catcode`\#12\relax}%
\@ifxundefined \pdfoutput {\@firstoftwo}{%
 \@ifnum{\z@=\pdfoutput}{\@firstoftwo}{\@secondoftwo}%
}{%
 \providecommand\@@startlink[1]{\leavevmode}%
 \providecommand\@@endlink[0]{}%
}{%
 \providecommand\@@startlink[1]{%
  \leavevmode
  \pdfstartlink
   attr{/Border[0 0 1 ]/H/I/C[0 1 1]}%
   user{/Subtype/Link/A<</Type/Action/S/URI/URI(#1)>>}%
  \relax
 }%
 \providecommand\@@endlink[0]{\pdfendlink}%
}%
\providecommand \url  [0]{\begingroup\@sanitize \@url }%
\providecommand \@url [1]{\endgroup\@href {#1}{\urlprefix}}%
\providecommand \urlprefix [0]{URL }%
\providecommand \Eprint[0]{\href }%
\@ifxundefined \urlstyle {%
  \providecommand \doi [1]{doi:\discretionary{}{}{}#1}%
}{%
  \providecommand \doi [0]{doi:\discretionary{}{}{}\begingroup
  \urlstyle{rm}\Url }%
}%
\providecommand \doibase [0]{http://dx.doi.org/}%
\providecommand \Doi[1]{\href{\doibase#1}}%
\providecommand \bibAnnote [3]{%
  \BibitemShut{#1}%
  \begin{quotation}\noindent
    \textsc{Key:}\ #2\\\textsc{Annotation:}\ #3%
  \end{quotation}%
}%
\providecommand \bibAnnoteFile [2]{%
  \IfFileExists{#2}{\bibAnnote {#1} {#2} {\input{#2}}}{}%
}%
\providecommand \typeout [0]{\immediate \write \m@ne }%
\providecommand \selectlanguage [0]{\@gobble}%
\providecommand \bibinfo [0]{\@secondoftwo}%
\providecommand \bibfield [0]{\@secondoftwo}%
\providecommand \translation [1]{[#1]}%
\providecommand \BibitemOpen[0]{}%
\providecommand \bibitemStop [0]{}%
\providecommand \bibitemNoStop [0]{.\EOS\space}%
\providecommand \EOS [0]{\spacefactor3000\relax}%
\providecommand \BibitemShut [1]{\csname bibitem#1\endcsname}%
%</preamble>
\bibitem{BB84}%
  \BibitemOpen
  \bibfield{author}{%
  \bibinfo {author} {\bibfnamefont{C.~H.}\ \bibnamefont{Bennett}}\ and\
  \bibinfo {author} {\bibfnamefont{G.}~\bibnamefont{Brassard}},\ }%
  in\ \emph{\bibinfo {booktitle} {Proc. IEEE Int. Conf. on Comp., Sys. and
  Signal Process.}}\ (\bibinfo {publisher} {IEEE},\ \bibinfo {address}
  {Bangalore},\ \bibinfo {year} {1984})\ pp.\ \bibinfo {pages} {175--179}%
  \bibAnnoteFile{NoStop}{BB84}%
\bibitem{E91}%
  \BibitemOpen
  \bibfield{author}{%
  \bibinfo {author} {\bibfnamefont{A.~K.}\ \bibnamefont{Ekert}},\ }%
  \bibfield{journal}{%
  \bibinfo {journal} {Phys. Rev. Lett.}\ }%
  \textbf{\bibinfo {volume} {67}},\ \bibinfo {pages} {661} (\bibinfo {month}
  {Aug}\ \bibinfo {year} {1991})%
  \bibAnnoteFile{NoStop}{E91}%
\bibitem{Note1}%
  \BibitemOpen
  \bibinfo {note} {An explicit calculation that shows that a constant-length
  initial key is sufficient to generate arbitrarily many novel key bits is
  given, for example, in~\cite {muller09}.}%
  \bibAnnoteFile{Stop}{Note1}%
\bibitem{Blum82}%
  \BibitemOpen
  \bibfield{author}{%
  \bibinfo {author} {\bibfnamefont{M.}~\bibnamefont{Blum}},\ }%
  \bibfield{journal}{%
  \bibinfo {journal} {SIGACT News}\ }%
  \textbf{\bibinfo {volume} {15}},\ \bibinfo {pages} {23} (\bibinfo {year}
  {1983})%
  \bibAnnoteFile{NoStop}{Blum82}%
\bibitem{HMU06}%
  \BibitemOpen
  \bibfield{author}{%
  \bibinfo {author} {\bibfnamefont{D.}~\bibnamefont{Hofheinz}}, \bibinfo
  {author} {\bibfnamefont{J.}~\bibnamefont{M{\"u}ller-Quade}},\ and\ \bibinfo
  {author} {\bibfnamefont{D.}~\bibnamefont{Unruh}},\ }%
  in\ \emph{\bibinfo {booktitle} {EUROCRYPT}},\ \bibinfo {series} {Lecture
  Notes in Computer Science}, Vol.\ \bibinfo {volume} {4004},\ \bibinfo
  {editor} {edited by\ \bibinfo {editor}
  {\bibfnamefont{S.}~\bibnamefont{Vaudenay}}}\ (\bibinfo {publisher}
  {Springer},\ \bibinfo {year} {2006})\ pp.\ \bibinfo {pages} {504--521}%
  \bibAnnoteFile{NoStop}{HMU06}%
\bibitem{Mayers97}%
  \BibitemOpen
  \bibfield{author}{%
  \bibinfo {author} {\bibfnamefont{D.}~\bibnamefont{Mayers}},\ }%
  \bibfield{journal}{%
  \bibinfo {journal} {Physical Review Letters}\ }%
  \textbf{\bibinfo {volume} {78}},\ \bibinfo {pages} {3414} (\bibinfo {year}
  {1997})%
  \bibAnnoteFile{NoStop}{Mayers97}%
\bibitem{LoChau97}%
  \BibitemOpen
  \bibfield{author}{%
  \bibinfo {author} {\bibfnamefont{H.~K.}\ \bibnamefont{Lo}}\ and\ \bibinfo
  {author} {\bibfnamefont{H.~F.}\ \bibnamefont{Chau}},\ }%
  \bibfield{journal}{%
  \bibinfo {journal} {Physical Review Letters}\ }%
  \textbf{\bibinfo {volume} {78}},\ \bibinfo {pages} {3410} (\bibinfo {year}
  {1997})%
  \bibAnnoteFile{NoStop}{LoChau97}%
\bibitem{DKSW06}%
  \BibitemOpen
  \bibfield{author}{%
  \bibinfo {author} {\bibfnamefont{G.~M.}\ \bibnamefont{D'Ariano}}, \bibinfo
  {author} {\bibfnamefont{D.}~\bibnamefont{Kretschmann}}, \bibinfo {author}
  {\bibfnamefont{D.}~\bibnamefont{Schlingemann}},\ and\ \bibinfo {author}
  {\bibfnamefont{R.~F.}\ \bibnamefont{Werner}},\ }%
  \bibfield{journal}{%
  \Doi{10.1103/PhysRevA.76.032328}{\bibinfo {journal} {Phys. Rev. A}}\ }%
  \textbf{\bibinfo {volume} {76}},\ \bibinfo {pages} {032328} (\bibinfo {month}
  {Sep}\ \bibinfo {year} {2007})%
  \bibAnnoteFile{NoStop}{DKSW06}%
\bibitem{CDPSW09}%
  \BibitemOpen
  \bibfield{author}{%
  \bibinfo {author} {\bibfnamefont{G.}~\bibnamefont{{Chiribella}}}, \bibinfo
  {author} {\bibfnamefont{G.~M.}\ \bibnamefont{{D'Ariano}}}, \bibinfo {author}
  {\bibfnamefont{P.}~\bibnamefont{{Perinotti}}}, \bibinfo {author}
  {\bibfnamefont{D.~M.}\ \bibnamefont{{Schlingemann}}},\ and\ \bibinfo {author}
  {\bibfnamefont{R.~F.}\ \bibnamefont{{Werner}}},\ }%
  \bibfield{journal}{%
  \bibinfo {journal} {ArXiv e-prints}}%
   (\bibinfo {month} {May}\ \bibinfo {year} {2009}),\
  \Eprint{http://arxiv.org/abs/0905.3801}{arXiv:0905.3801 [quant-ph]}%
  \bibAnnoteFile{NoStop}{CDPSW09}%
\bibitem{GMR85}%
  \BibitemOpen
  \bibfield{author}{%
  \bibinfo {author} {\bibfnamefont{S.}~\bibnamefont{Goldwasser}}, \bibinfo
  {author} {\bibfnamefont{S.}~\bibnamefont{Micali}},\ and\ \bibinfo {author}
  {\bibfnamefont{C.}~\bibnamefont{Rackoff}},\ }%
  in\ \emph{\bibinfo {booktitle} {STOC}}\ (\bibinfo {publisher} {ACM},\
  \bibinfo {year} {1985})\ pp.\ \bibinfo {pages} {291--304}%
  \bibAnnoteFile{NoStop}{GMR85}%
\bibitem{BBCS92}%
  \BibitemOpen
  \bibfield{author}{%
  \bibinfo {author} {\bibfnamefont{C.~H.}\ \bibnamefont{Bennett}}, \bibinfo
  {author} {\bibfnamefont{G.}~\bibnamefont{Brassard}}, \bibinfo {author}
  {\bibfnamefont{C.}~\bibnamefont{Cr{\'e}peau}},\ and\ \bibinfo {author}
  {\bibfnamefont{H.}~\bibnamefont{Skubiszewska}},\ }%
  in\ \emph{\bibinfo {booktitle} {Advances in Cryptology --- CRYPTO '91}},\
  \bibinfo {series} {Lecture Notes in Computer Science}, Vol.\ \bibinfo
  {volume} {576}\ (\bibinfo {publisher} {Springer},\ \bibinfo {year} {1992})\
  pp.\ \bibinfo {pages} {351--366}%
  \bibAnnoteFile{NoStop}{BBCS92}%
\bibitem{DFLSS09}%
  \BibitemOpen
  \bibfield{author}{%
  \bibinfo {author} {\bibfnamefont{I.}~\bibnamefont{Damg{\aa}rd}}, \bibinfo
  {author} {\bibfnamefont{S.}~\bibnamefont{Fehr}}, \bibinfo {author}
  {\bibfnamefont{C.}~\bibnamefont{Lunemann}}, \bibinfo {author}
  {\bibfnamefont{L.}~\bibnamefont{Salvail}},\ and\ \bibinfo {author}
  {\bibfnamefont{C.}~\bibnamefont{Schaffner}},\ }%
  in\ \emph{\bibinfo {booktitle} {CRYPTO}},\ \bibinfo {series} {Lecture Notes
  in Computer Science}, Vol.\ \bibinfo {volume} {5677},\ \bibinfo {editor}
  {edited by\ \bibinfo {editor} {\bibfnamefont{S.}~\bibnamefont{Halevi}}}\
  (\bibinfo {publisher} {Springer},\ \bibinfo {year} {2009})\ pp.\ \bibinfo
  {pages} {408--427}%
  \bibAnnoteFile{NoStop}{DFLSS09}%
\bibitem{Unruh10}%
  \BibitemOpen
  \bibfield{author}{%
  \bibinfo {author} {\bibfnamefont{D.}~\bibnamefont{Unruh}},\ }%
  in\ \emph{\bibinfo {booktitle} {EUROCRYPT}},\ \bibinfo {series} {Lecture
  Notes in Computer Science}, Vol.\ \bibinfo {volume} {6110},\ \bibinfo
  {editor} {edited by\ \bibinfo {editor}
  {\bibfnamefont{H.}~\bibnamefont{Gilbert}}}\ (\bibinfo {publisher}
  {Springer},\ \bibinfo {year} {2010})\ pp.\ \bibinfo {pages} {486--505}%
  \bibAnnoteFile{NoStop}{Unruh10}%
\bibitem{Kilian88}%
  \BibitemOpen
  \bibfield{author}{%
  \bibinfo {author} {\bibfnamefont{J.}~\bibnamefont{Kilian}},\ }%
  in\ \emph{\bibinfo {booktitle} {Proceedings of the 20th Annual ACM Symposium
  on Theory of Computing (STOC~'88)}}\ (\bibinfo {publisher} {ACM Press},\
  \bibinfo {year} {1988})\ pp.\ \bibinfo {pages} {20--31}%
  \bibAnnoteFile{NoStop}{Kilian88}%
\bibitem{WW10}%
  \BibitemOpen
  \bibfield{author}{%
  \bibinfo {author} {\bibfnamefont{S.}~\bibnamefont{Winkler}}\ and\ \bibinfo
  {author} {\bibfnamefont{J.}~\bibnamefont{Wullschleger}},\ }%
  in\ \emph{\bibinfo {booktitle} {CRYPTO}},\ \bibinfo {series} {Lecture Notes
  in Computer Science}, Vol.\ \bibinfo {volume} {6223},\ \bibinfo {editor}
  {edited by\ \bibinfo {editor} {\bibfnamefont{T.}~\bibnamefont{Rabin}}}\
  (\bibinfo {publisher} {Springer},\ \bibinfo {year} {2010})\ pp.\ \bibinfo
  {pages} {707--723}%
  \bibAnnoteFile{NoStop}{WW10}%
\bibitem{Note2}%
  \BibitemOpen
  \bibinfo {note} {Using the equivalence of oblivious transfer and commitments,
  the result of~\cite {WW10} implies that there exists no \protect \emph
  {composable} protocol that implements $(m+1)$ \protect \emph {individual} bit
  commitments using $m$ bit commitments as a resource, if one demands that the
  error decreases exponentially in $m$.}%
  \bibAnnoteFile{Stop}{Note2}%
\bibitem{SR01}%
  \BibitemOpen
  \bibfield{author}{%
  \bibinfo {author} {\bibfnamefont{R.~W.}\ \bibnamefont{Spekkens}}\ and\
  \bibinfo {author} {\bibfnamefont{T.}~\bibnamefont{Rudolph}},\ }%
  \bibfield{journal}{%
  \Doi{10.1103/PhysRevA.65.012310}{\bibinfo {journal} {Phys. Rev. A}}\ }%
  \textbf{\bibinfo {volume} {65}},\ \bibinfo {pages} {012310} (\bibinfo {month}
  {Dec}\ \bibinfo {year} {2001})%
  \bibAnnoteFile{NoStop}{SR01}%
\bibitem{BCHLW06}%
  \BibitemOpen
  \bibfield{author}{%
  \bibinfo {author} {\bibfnamefont{H.}~\bibnamefont{Buhrman}}, \bibinfo
  {author} {\bibfnamefont{M.}~\bibnamefont{Christandl}}, \bibinfo {author}
  {\bibfnamefont{P.}~\bibnamefont{Hayden}}, \bibinfo {author}
  {\bibfnamefont{H.-K.}\ \bibnamefont{Lo}},\ and\ \bibinfo {author}
  {\bibfnamefont{S.}~\bibnamefont{Wehner}},\ }%
  \bibfield{journal}{%
  \Doi{10.1103/PhysRevLett.97.250501}{\bibinfo {journal} {Phys. Rev. Lett.}}\
  }%
  \textbf{\bibinfo {volume} {97}},\ \bibinfo {pages} {250501} (\bibinfo {month}
  {Dec}\ \bibinfo {year} {2006})%
  \bibAnnoteFile{NoStop}{BCHLW06}%
\bibitem{CK11}%
  \BibitemOpen
  \bibfield{author}{%
  \bibinfo {author} {\bibfnamefont{A.}~\bibnamefont{{Chailloux}}}\ and\
  \bibinfo {author} {\bibfnamefont{I.}~\bibnamefont{{Kerenidis}}},\ }%
  \bibfield{journal}{%
  \bibinfo {journal} {ArXiv e-prints}}%
   (\bibinfo {month} {Feb.}\ \bibinfo {year} {2011}),\
  \Eprint{http://arxiv.org/abs/1102.1678}{arXiv:1102.1678 [quant-ph]}%
  \bibAnnoteFile{NoStop}{CK11}%
\bibitem{Kent11}%
  \BibitemOpen
  \bibfield{author}{%
  \bibinfo {author} {\bibfnamefont{A.}~\bibnamefont{Kent}},\ }%
  \enquote{\bibinfo {title} {Unconditionally secure bit commitment with flying
  qudits},}\  (\bibinfo {year} {2011}),\
  \Eprint{http://arxiv.org/abs/arXiv:1101.4620}{arXiv:1101.4620}%
  \bibAnnoteFile{NoStop}{Kent11}%
\bibitem{Crepea97}%
  \BibitemOpen
  \bibfield{author}{%
  \bibinfo {author} {\bibfnamefont{C.}~\bibnamefont{Cr{\'e}peau}},\ }%
  in\ \emph{\bibinfo {booktitle} {Advances in Cryptology --- CRYPTO '97}},\
  \bibinfo {series} {Lecture Notes in Computer Science}, Vol.\ \bibinfo
  {volume} {1233}\ (\bibinfo {publisher} {Springer},\ \bibinfo {year} {1997})\
  pp.\ \bibinfo {pages} {306--317}%
  \bibAnnoteFile{NoStop}{Crepea97}%
\bibitem{IMNW04}%
  \BibitemOpen
  \bibfield{author}{%
  \bibinfo {author} {\bibfnamefont{H.}~\bibnamefont{Imai}}, \bibinfo {author}
  {\bibfnamefont{J.}~\bibnamefont{M{\"u}ller-Quade}}, \bibinfo {author}
  {\bibfnamefont{A.}~\bibnamefont{Nascimento}},\ and\ \bibinfo {author}
  {\bibfnamefont{A.}~\bibnamefont{Winter}},\ }%
  in\ \emph{\bibinfo {booktitle} {Proceedings of the IEEE International
  Symposium on Information Theory (ISIT~'04)}}\ (\bibinfo {year} {2004})%
  \bibAnnoteFile{NoStop}{IMNW04}%
\bibitem{WolWul04}%
  \BibitemOpen
  \bibfield{author}{%
  \bibinfo {author} {\bibfnamefont{S.}~\bibnamefont{Wolf}}\ and\ \bibinfo
  {author} {\bibfnamefont{J.}~\bibnamefont{Wullschleger}},\ }%
  in\ \emph{\bibinfo {booktitle} {Proceedings of 2004 IEEE Information Theory
  Workshop (ITW~'04)}}\ (\bibinfo {year} {2004})%
  \bibAnnoteFile{NoStop}{WolWul04}%
\bibitem{Note3}%
  \BibitemOpen
  \bibinfo {note} {This assumption is not justified in the relativistic setting
  considered in~\cite {Kent11}.}%
  \bibAnnoteFile{Stop}{Note3}%
\bibitem{Note4}%
  \BibitemOpen
  \bibinfo {note} {In particular, one would have to consider arbitrary
  malicious strategies of dishonest parties to prove the security of a
  protocol.}%
  \bibAnnoteFile{Stop}{Note4}%
\bibitem{DMS00}%
  \BibitemOpen
  \bibfield{author}{%
  \bibinfo {author} {\bibfnamefont{P.}~\bibnamefont{Dumais}}, \bibinfo {author}
  {\bibfnamefont{D.}~\bibnamefont{Mayers}},\ and\ \bibinfo {author}
  {\bibfnamefont{L.}~\bibnamefont{Salvail}},\ }%
  in\ \emph{\bibinfo {booktitle} {EUROCRYPT}},\ \bibinfo {series} {Lecture
  Notes in Computer Science}, Vol.\ \bibinfo {volume} {1807},\ \bibinfo
  {editor} {edited by\ \bibinfo {editor}
  {\bibfnamefont{B.}~\bibnamefont{Preneel}}}\ (\bibinfo {publisher} {LNCS},\
  \bibinfo {year} {2000})\ pp.\ \bibinfo {pages} {300--315}%
  \bibAnnoteFile{NoStop}{DMS00}%
\bibitem{Note5}%
  \BibitemOpen
  \bibinfo {note} {See EPAPS Document No.[number will be inserted by
  publisher].}%
  \bibAnnoteFile{Stop}{Note5}%
\bibitem{KRS09}%
  \BibitemOpen
  \bibfield{author}{%
  \bibinfo {author} {\bibfnamefont{R.}~\bibnamefont{K\"onig}}, \bibinfo
  {author} {\bibfnamefont{R.}~\bibnamefont{Renner}},\ and\ \bibinfo {author}
  {\bibfnamefont{C.}~\bibnamefont{Schaffner}},\ }%
  \bibfield{journal}{%
  \Doi{10.1109/TIT.2009.2025545}{\bibinfo {journal} {Information Theory, IEEE
  Transactions on}}\ }%
  \textbf{\bibinfo {volume} {55}},\ \bibinfo {pages} {4337 } (\bibinfo {month}
  {sept.}\ \bibinfo {year} {2009}),\ ISSN \bibinfo {issn} {0018-9448}%
  \bibAnnoteFile{NoStop}{KRS09}%
\bibitem{TCR10}%
  \BibitemOpen
  \bibfield{author}{%
  \bibinfo {author} {\bibfnamefont{M.}~\bibnamefont{Tomamichel}}, \bibinfo
  {author} {\bibfnamefont{R.}~\bibnamefont{Colbeck}},\ and\ \bibinfo {author}
  {\bibfnamefont{R.}~\bibnamefont{Renner}},\ }%
  \bibfield{journal}{%
  \bibinfo {journal} {IEEE Transactions on Information Theory}\ }%
  \textbf{\bibinfo {volume} {56}},\ \bibinfo {pages} {4674 } (\bibinfo {year}
  {2010})%
  \bibAnnoteFile{NoStop}{TCR10}%
\bibitem{Renner2005}%
  \BibitemOpen
  \bibfield{author}{%
  \bibinfo {author} {\bibfnamefont{R.}~\bibnamefont{Renner}},\ }%
  \bibinfo {journal} {Ph.D. thesis, ETH Zurich, (2005), arXiv:
  quant-ph/0512258}%
  \bibAnnoteFile{NoStop}{Renner2005}%
\bibitem{Tom2010}%
  \BibitemOpen
\bibfield{journal}{%
    }%
  \bibfield{author}{%
  \bibinfo {author} {\bibfnamefont{M.}~\bibnamefont{Tomamichel}}, \bibinfo
  {author} {\bibfnamefont{R.}~\bibnamefont{Renner}}, \bibinfo {author}
  {\bibfnamefont{C.}~\bibnamefont{Schaffner}},\ and\ \bibinfo {author}
  {\bibfnamefont{A.}~\bibnamefont{Smith}},\ }%
  in\ \emph{\bibinfo {booktitle} {Information Theory Proceedings (ISIT), 2010
  IEEE International Symposium on}}\ (\bibinfo {year} {2010})\ pp.\ \bibinfo
  {pages} {2703 --2707}%
  \bibAnnoteFile{NoStop}{Tom2010}%
\bibitem{Lo97}%
  \BibitemOpen
  \bibfield{author}{%
  \bibinfo {author} {\bibfnamefont{H.~K.}\ \bibnamefont{Lo}},\ }%
  \bibfield{journal}{%
  \bibinfo {journal} {Physical Review A}\ }%
  \textbf{\bibinfo {volume} {56}},\ \bibinfo {pages} {1154} (\bibinfo {year}
  {1997})%
  \bibAnnoteFile{NoStop}{Lo97}%
\bibitem{BW92}%
  \BibitemOpen
  \bibfield{author}{%
  \bibinfo {author} {\bibfnamefont{C.~H.}\ \bibnamefont{Bennett}}\ and\
  \bibinfo {author} {\bibfnamefont{S.~J.}\ \bibnamefont{Wiesner}},\ }%
  \bibfield{journal}{%
  \Doi{10.1103/PhysRevLett.69.2881}{\bibinfo {journal} {Phys. Rev. Lett.}}\ }%
  \textbf{\bibinfo {volume} {69}},\ \bibinfo {pages} {2881} (\bibinfo {month}
  {Nov}\ \bibinfo {year} {1992})%
  \bibAnnoteFile{NoStop}{BW92}%
\bibitem{muller09}%
  \BibitemOpen
  \bibfield{author}{%
  \bibinfo {author} {\bibfnamefont{J.}~\bibnamefont{M\"{u}ller-Quade}}\ and\
  \bibinfo {author} {\bibfnamefont{R.}~\bibnamefont{Renner}},\ }%
  \bibfield{journal}{%
  \bibinfo {journal} {New J. Phys.}\ }%
  \textbf{\bibinfo {volume} {11}},\ \bibinfo {pages} {085006} (\bibinfo {month}
  {Aug.}\ \bibinfo {year} {2009})%
  \bibAnnoteFile{NoStop}{muller09}%
\bibitem{FG99}%
  \BibitemOpen
  \bibfield{author}{%
  \bibinfo {author} {\bibfnamefont{C.~A.}\ \bibnamefont{Fuchs}}\ and\ \bibinfo
  {author} {\bibfnamefont{J.}~\bibnamefont{van~de Graaf}},\ }%
  \bibfield{journal}{%
  \bibinfo {journal} {IEEE Transactions on Information Theory}\ }%
  \textbf{\bibinfo {volume} {45}},\ \bibinfo {pages} {1216} (\bibinfo {year}
  {1999})%
  \bibAnnoteFile{NoStop}{FG99}%
\bibitem{CW79}%
  \BibitemOpen
  \bibfield{author}{%
  \bibinfo {author} {\bibfnamefont{J.~L.}\ \bibnamefont{Carter}}\ and\ \bibinfo
  {author} {\bibfnamefont{M.~N.}\ \bibnamefont{Wegman}},\ }%
  \bibfield{journal}{%
  \bibinfo {journal} {Journal of Computer and System Sciences}\ }%
  \textbf{\bibinfo {volume} {18}},\ \bibinfo {pages} {143} (\bibinfo {year}
  {1979})%
  \bibAnnoteFile{NoStop}{CW79}%
\end{thebibliography}
\end{document}